\documentclass{IEEEojcsys}

\usepackage[colorlinks,urlcolor=blue,linkcolor=blue,citecolor=blue]{hyperref}
\usepackage[usenames,dvipsnames]{xcolor}

\usepackage{graphicx}

\usepackage{orcidlink}

\jvol{00}
\jnum{XX}
\paper{1234567}
\pubyear{2021}
\receiveddate{XX September 2021}
\accepteddate{XX October 2021}
\publisheddate{XX November 2021}
\currentdate{XX November 2021}
\doiinfo{OJCSYS.2021.Doi Number}

\graphicspath{ {./img/} }

\newcommand{\intra}{{\rm intra}} %
\newcommand{\ctr}{{\rm ctl}} %

\newcommand{\EP}{{\frac{1}{\varepsilon}}} %
\newcommand{\EPt}{{\frac{t}{\varepsilon}}} %

\newcommand\PC[1]{^{(#1)}}

\newcommand\scalemath[2]{\scalebox{#1}{\mbox{\ensuremath{\displaystyle #2}}}}

\newcommand{\blkdiag}{{\rm blkdiag}} %

\newcommand{\bl}{\color{blue}}

\definecolor{ORG}{HTML}{fe9929}

\usepackage{mathtools}
\usepackage{bm}
\usepackage{bbm}

\usepackage{graphicx}
\usepackage{amssymb}

\usepackage{algorithm}
\usepackage{algpseudocode}

\newcommand{\alias}{/Users/yuzhen.qin/Library/CloudStorage/Dropbox/bib/alias}
\newcommand{\New}{/Users/yuzhen.qin/Library/CloudStorage/Dropbox/bib/New}
\newcommand{\Main}{/Users/yuzhen.qin/Library/CloudStorage/Dropbox/bib/Main}
\newcommand{\FP}{/Users/yuzhen.qin/Library/CloudStorage/Dropbox/bib/FP}

\newcommand{\BE}{\mathbb{E}} 
\newcommand{\BS}{\mathbb{S}} 
\newcommand{\BG}{\mathbb{G}} 
\newcommand{\R}{\mathbb{R}} 
\newcommand{\BONE}{\mathbbm{1}} 
\newcommand{\BT}{\mathbb{T}} 
\newcommand{\BV}{\mathbb{V}} 
\newcommand{\BB}{\mathbb{B}} 

\newcommand{\BmW}{\bm{W}}
\newcommand{\BmV}{\bm{V}}
\newcommand{\BmZr}{\bm{0}}
\newcommand{\BmU}{\bm{U}}

\newcommand{\CP}{\mathcal{P}} %
\newcommand{\CR}{\mathcal{R}} %
\newcommand{\CC}{\mathcal{C}} %
\newcommand{\CG}{\mathcal{G}} %
\newcommand{\CE}{\mathcal{E}} %
\newcommand{\CM}{\mathcal{M}} %
\newcommand{\CV}{\mathcal{V}} %
\newcommand{\CD}{\mathcal{D}} %

\newcommand{\sign}{{\rm sign}}
\newcommand{\inter}{{\rm inter}}

\newcommand{\ctl}{{\rm ctl}}

\newcommand{\diag}{{\rm diag}} %

\usepackage{amsthm}

\newtheorem{theorem}{Theorem}

\newtheorem{lemma}{Lemma}
\newtheorem{corollary}{Corollary}

\theoremstyle{definition}
\newtheorem{remark}{Remark}

\newtheorem{definition}{Definition}
\newtheorem{assumption}{Assumption}

\newcommand*{\QE}{\hfill\ensuremath{\blacksquare}}	

\newcommand*{\QEDA}{\null \hfill\ensuremath{\triangle}}

\begin{document}

\sptitle{Article Category}

\title{Vibrational Stabilization of Cluster Synchronization in Oscillator Networks} 

\editor{This paper was recommended by Associate Editor F. A. Author.}

\author{Yuzhen Qin\affilmark{1}\orcidlink{0000-0003-1851-1370} (Member, IEEE)}

\author{Alberto Maria Nobili\affilmark{2}} 

\author{Danielle S. Bassett\affilmark{3}\orcidlink{0000-0002-6183-4493}  (Member, IEEE)}

\author{Fabio Pasqualetti\affilmark{4}\orcidlink{0000-0002-8457-8656} (Member, IEEE)}

\affil{Department of Artificial Intelligence, Donders Institute for Brain, Cognition, and Behaviour, Radboud University, Nijmegen, the Netherlands}

\affil{Perceptual Robotics Laboratory at the IIM Institute, Department of Excellence in Robotics and A.I., Scuola Superiore Sant’Anna, Pisa, Italy} 

\affil{ Department of
	Bioengineering, Department of Electrical \& Systems
	Engineering,  Department of Physics \& Astronomy, the
	Department of Psychiatry, and  Department of Neurology,
	University of Pennsylvania, and the Santa Fe Institute}
	
\affil{ Department of Mechanical Engineering, University of
	California, Riverside, CA, USA} 

\corresp{CORRESPONDING AUTHOR: Yuzhen Qin (e-mail: \href{mailto:yuzhen.qin@donders.ru.nl}{yuzhen.qin@donders.ru.nl})}
\authornote{This research was funded in part by Aligning Science Across	Parkinson’s (ASAP-020616) and in part by NSF (NCS-FO-1926829) and ARO (W911NF1910360).
	Qin's work in also funded in part by the project Dutch Brain Interface Initiative (DBI2) with project number 024.005.022 of the research programme Gravitation which is (partly) financed by the Dutch Research Council (NWO). 
	 Y. Qin, F. Pasqualetti and D. S. Bassett  were also with  Aligning Science Across Parkinson’s (ASAP) Collaborative Research Network, Chevy Chase, MD 20815. The manuscript was prepared using LaTex on Overleaf, and the source code can be downloaded at the Arxiv page. All simulations were performed in MATLAB, and the corresponding code is available at this repository: \hyperref[]{https://doi.org/10.5281/zenodo.17170702}}

\markboth{PREPARATION OF PAPERS FOR IEEE OPEN JOURNAL OF CONTROL SYSTEMS}{F. A. AUTHOR {\itshape ET AL}.}

\begin{abstract}
	Cluster synchronization is of great importance for the normal functioning of numerous technological and natural systems. Deviations from normal cluster synchronization patterns are closely associated with various malfunctions, such as neurological disorders in the brain. Therefore, it is crucial to restore normal system functions by stabilizing the appropriate cluster synchronization patterns. Most existing studies focus on designing controllers based on state measurements to achieve system stabilization. However, in many real-world scenarios, measuring system states in real time, such as neuronal activity in the brain, poses significant challenges, rendering the stabilization of such systems difficult. To overcome this challenge, in this paper, we employ an open-loop control strategy, \textit{vibrational control}, which does not require any state measurements. We establish some sufficient conditions under which vibrational inputs stabilize cluster synchronization. Further, we provide a tractable approach to design vibrational control. Finally, numerical experiments are conducted to demonstrate our theoretical findings.
\end{abstract}


\maketitle

\section{Introduction}
Cluster synchronization describes the phenomenon in which units within a network exhibit synchronized behavior, forming distinct clusters. This intriguing phenomenon is widely observed in many natural and engineering systems. For instance, it manifests as correlated neural activity in the brain. Different patterns of cluster synchronization in the brain play a fundamental role in various functions, including neuronal communication, memory formation and retrieval, and motor function \cite{HG-PAA-DG-AA-KA:19,FJ-AN:2011}.  

However, many brain disorders, such as Parkinson's disease \cite{Hammond2007} and epilepsy \cite{JP-DCM-JJGR-SCA:2013}, are characterized by aberrant synchrony patterns of brain activity. It becomes crucial to be able to stabilize normal patterns of cluster synchronization. Most existing studies rely on the assumption that the states of systems can be measured to design feedback controllers to stabilize them. Unfortunately, many real-world systems often exhibit complex nonlinear dynamics over large network structures with states that are difficult to observe or measure in real time (such as neuronal activity in the brain), thus preventing the use of
sophisticated feedback techniques.

In this paper we leverage \emph{vibrational control} strategies to
ensure stability of network systems. Vibrational control is a powerful
strategy applicable in various domains to stabilize the dynamics of
complex systems without the need for direct state measurements. Unlike
traditional feedback-based control methods that rely on directly
measuring the system's states or outputs, vibrational control
leverages the inherent dynamics of the system to induce desired
stability and performance. In particular, vibrational control uses
pre-designed high-frequency signals, injected at specific locations
and times. As we show below, these signals can effectively change the
system dynamics and suppress unstable dynamics. Successful
applications of vibrational control to individual systems are
numerous, including inverted pendulums, chemical reactors, and
under-actuated robots \cite{REB-JB-SMM:86b,BS-BTZ:97,CX-TY-MI:2018}.
This paper takes the first steps to develop a theory of vibrational
control for network systems. Interestingly, in addition to its
technological value, the theory of vibrational control for network
systems may also help explain the success of deep brain stimulation
methods, as this technique also relies on the injection of
high-frequency electrical pulses to regulate brain processes and
restore healthy functions \cite{Krauss2021}.

We study the stabilization of networks of heterogeneous Kuramoto
oscillators, and in particular on the stabilization of the dynamics
around a desired cluster synchronization manifold
\cite{TM-GB-DSB-FP:18}. We remark that (i) Kuramoto-oscillator networks have been used successfully to model different phenomena
in diverse domains ranging from power engineering to biology and
neuroscience, thus making our results of broad applicability, and (ii)
cluster synchronization has been used as a proxy to model and regulate
the emergence of functional activation patterns in the brain
\cite{TM-GB-DSB-FP:19b,TM-GB-DSB-FP:22}, thus making our results of
timely relevance to these problems. 

\noindent
\textbf{Related work.}
Cluster synchronization has garnered significant attention recently as researchers seek to understand its underlying mechanisms and control strategies. Investigations into the field have revealed intriguing connections between cluster synchronization and network symmetries \cite{Pecora2014, YSC-TN-AEM:17, YQ-MC-BDOA-DSB-FP:20,EJ-SA-SJ-CJP-DRM:20} as well as equitable partitions \cite{Schaub2016}. Furthermore, stability conditions for cluster synchronization have been established in networks featuring dyadic connections \cite{TM-GB-DSB-FP:18, QY-KY-PO-CM:21, FP-SA-MT:21, FP-SA-MT:2020, KR-IH:2021} and hyper connections \cite{SA-DRM:2021a, SA-DRM:2021b}.
To control cluster synchronization, researchers have proposed diverse strategies, such as pinning control \cite{WW-WZ-TC:09} and interventions that involve manipulating network connections or the dynamics of individual nodes \cite{GLV-FM-LV:2018, DF-GR-MDB:17,AA-FC-FP-JC:22}. 
In contrast, our approach focuses on vibrational control, which offers a more realistic strategy in many real-world systems, e.g., for regulating neural activity as it resembles deep brain stimulation \cite{Krauss2021}. To the best of our knowledge, our work is among the first ones to utilize vibrations to regulate network systems.

\noindent
\textbf{Paper contributions.} The main contributions of this paper are
as follows. First, we formalize the problem of vibrational control for Kuramoto-oscillator networks, with the aim to stabilize patterns of cluster synchronization. By employing an averaging technique, we demonstrate that introducing vibrational inputs into the network effectively alters the system dynamics on average. We establish sufficient conditions for the effectiveness of vibrational control in stabilizing cluster synchronization within Kuramoto-oscillator networks. Through the analysis of linearized systems, we gain deep insights into the underlying mechanisms of vibrational control, revealing its ability to enhance the robustness of synchronization within clusters. Second, we establish a connection between the design of vibrational control in Kuramoto-oscillator networks and linear network systems. We show that vibrational control can effectively modify the weights of edges in linear network systems, thereby shaping their robustness. We develop graph-theoretical conditions for determining which edges can be modified through vibration. Additionally, we propose a systematic approach to designing vibrational control that targets the modifiable edges, aiming to enhance the overall robustness of the network system. Building on these findings, we apply the results to homogeneous Kuramoto-oscillator networks, presenting a tractable approach to designing vibrational control for improving the robustness of full synchronization. Furthermore, we extend the application of these results to heterogeneous Kuramoto-oscillator networks, deriving precise forms and placements of vibrational inputs to stabilize cluster synchronization. Finally, we conduct a numerical experiment to demonstrate our method for designing vibrational control to stabilize cluster synchronization in Kuramoto-oscillator networks.

A preliminary version of this work appeared in \cite{YQ-DSB-FP:22a}. Compared with it, this paper presents a more comprehensive approach to design vibrational control in Kuramoto-oscillator networks by deriving and utilizing precise forms and placements of vibrational inputs in linear network systems.

\textbf{Notation.} Denote the unit circle by $\BS^1$, and a point on it is called a \textit{phase} since the point can be used to indicate the phase angle of an oscillator.  For any two phases $\theta_1,\theta_2 \in \BS^1$, the geodesic distance between them is the minimum of the lengths of the counter-clockwise and clockwise arcs connecting them, which is denoted by $|\theta_1-\theta_2 |_\BS$. Given a matrix $B\in\R^{m\times n}$, the matrix $A = \sign(B)$ is defined in a way such that $a_{ij}=1$ if $b_{ij}>0$, $a_{ij}=-1$ if $b_{ij}<0$, and $a_{ij}=0$ if $b_{ij}=0$. Given a matrix $A$, $A^\dagger$ denote its pseudo-inverse.  Given a vector  $a\in \R^n$, $\diag(a)$ denotes the diagonal matrix formed by $a$. Given matrices $A_1, A_2,\dots,A_n$, $\blkdiag(A_1, A_2, \dots, A_n)$ denotes the block diagonal matrix formed by them. Denote $\otimes$ and $\odot$ as the Kronecker product and point-wise product, respectively.

\section{Problem Formulation}
\subsection{Kuramoto-Oscillator Networks}

Consider a network of $n$ Kuramoto oscillators with dynamics described by 
\begin{equation}\label{no_input}
	\dot \theta_i =\omega_i + \sum_{j=1}^{n} w_{ij} \sin(\theta_j- \theta_i), 
\end{equation}
where $\theta_i \in \BS^1$ is the phase of the $i$th oscillator, $\omega_i\in \R$ is its natural frequency for $i=1,\dots,n$, and $w_{ij}$ is the coupling strength. This paper investigates a network configuration where connections between nodes are bidirectional but allow for asymmetry (i.e., $w_{ji}\neq 0$ if $w_{ij}\neq 0$, and $w_{ij}\neq w_{ji}$ is allowed). This assumption represents a more flexible scenario compared to the commonly studied undirected networks in existing literature. We use a weighted directed graph  $\CG=(\CV,\CE,W)$ to describe the network structure,  where $\CV=\{1,2,\dots, n\}$, $\CE\subseteq \CV\times \CV$, and $W=[w_{ij}]_{n\times n}$ with $w_{ij}\ge 0$ is the weighted adjacency matrix. There is a directed edge from $j$ to $i$ in $\CG$, i.e., $(j,i)\in \CE$, if $w_{ij}>0$; this edge does not exist, i.e., $(j,i)\notin \CE$,  if $w_{ij}=0$.


In this paper, we are interest in studying cluster synchronization in the Kuramoto-oscillator network \eqref{no_input}. Let us first formally define cluster synchronization. 

For the graph $\CG=(\CV,\CE,W)$, consider the partition of the nodes in it below:
\begin{equation}\label{partition}
	\CP:=\{\CP_1,\CP_2,\dots,\CP_r\},
\end{equation}
where each $\CP_k$ is a subset of $\CV$ satisfying $\CP_k \cap  \CP_\ell =\emptyset$ for any $k\neq \ell$, and $\cup_{k=1}^r \CP_k=\CV$. 

\begin{definition}[\textit{Cluster Synchronization Manifold}]\label{CS:manifold}
	The \textit{cluster synchronization manifold} with respect to the partition $\CP$ is defined as
	\begin{equation*}
		\CM:= \{\theta \in \BT^n: \theta_i=\theta_j, \forall i,j \in \CP_k,  k= 1,\dots,r\}.\hspace{12pt}\QEDA
	\end{equation*}
\end{definition} 

Note that  various cluster synchronization patterns may emerge in the same network, each of which corresponds to a distinct network partition. Nevertheless, the partition defined in \eqref{partition} exhibits sufficient generality to capture any cluster synchronization pattern, enabling the analysis of different patterns of cluster synchronization. 

The manifold $\CM$ is invariant along the system \eqref{no_input} if, starting from $\theta(0)\in \CM$, the solution to \eqref{no_input} satisfies $\theta(t)\in \CM$ for all $t\ge 0$. We make the following assumption to ensure the invariance of $\CM$.

\begin{assumption}[\textit{Invariance}]\label{invariance}
	For $k=1,2,\dots, r$:  i) the natural frequencies satisfy $\omega_i=\omega_j$ for any $i,j \in \CP_k$; and ii) the coupling strengths satisfy that, for any $\ell\in\{1,2,\dots,r\}\backslash\{k\}$,  $\sum_{q\in \CP_\ell}(w_{iq}-w_{jq})=0$ for any $i,j \in \CP_k$. \QEDA
\end{assumption}

This assumption guarantees that oscillators within a cluster receive identical inputs from every other  synchronized cluster, a critical factor for maintaining synchronization among them. Similar assumptions are made for undirected networks in earlier studies \cite{TM-GB-DSB-FP:18,LT-CF-MI-DSB-FP:17}. 
To ensure the emergence of cluster synchronization, it is essential to not only establish invariance but also ensure the stability of $\CM$. Given a manifold $\mathcal C \in \BT^n$, define a $\delta$-neighborhood of $\mathcal C$ by $U_{\delta}(\mathcal C) = \{\theta \in \BT^n:{\rm dist}(\theta,\mathcal C)<\delta\}$ with ${\rm dist}(\theta,\mathcal C)=\inf _{y\in \mathcal C}{\|\theta-y\|_\BS}$. The  exponential stability of the manifold $\CM$ is defined below. 

\begin{definition}
	The manifold $\CM\in \BT^n$ is said to be exponentially stable  along the system \eqref{no_input}  if there is $\delta>0$ such that for any initial phase $\theta(0) \in\BT^ n$ satisfying $\theta(0)\in U_{\delta}( \CM)$ it holds that for all $t\ge 0$, ${\rm dist}(\theta(t),\CM)=k\cdot{\rm dist}(\theta(0), \mathcal C)\cdot e^{-\lambda t}$ for some $k>0$ and $\lambda>0$.
\end{definition}

Sufficient conditions are constructed for the exponential stability of $\CM$ (e.g., see \cite{TM-GB-DSB-FP:18,QY-KY-PO-CM:21,Schaub2016}). However, variations in network parameters due to factors, e.g., aging or brain disorders, can disrupt such conditions, resulting in the loss of stability of cluster synchronization. This paper focuses on the application of \textit{vibrational control}, a control strategy reminiscent of deep brain stimulation (DBS) \cite{Krauss2021}, to restore the stability of desired cluster synchronization patterns. Next, we present an introduction to vibrational control.

\subsection{Vibrational Control}
Following \cite{REB-JB-SMM:86a}, consider a nonlinear system
\begin{equation}\label{general}
	\dot x = f(x, a),
\end{equation} 
where $x\in\R^n$, and $a\in \R^{m}$ represents the parameters of the system. For linear systems, $a$ can be rewritten into a matrix form $A\in \R^{n\times n}$, and then we have $\dot x =A x$.  Vibrational control introduces vibrations to the parameters of \eqref{general}, resulting in
\begin{equation}\label{general:ctred}
	\dot x = f(x, a+v(t)).
\end{equation} 
The control vector $v(t)=[v_1, v_2, \dots, v_m]^\top\in \mathbb{R}^m$ is usually selected to have the following structure:
\begin{equation}\label{vib:gener_form}
	v_{i}(t)= \sum_{\ell=1}^{\infty} \alpha_{i}^{(\ell)} \sin(\ell \beta_{i}t+\varphi_{i}^{(\ell)}).
\end{equation}
which is almost-periodic, zero-mean, and high-frequency \cite{REB-JB-SMM:86b}. Vibrational control is an open-loop strategy. An appropriate configuration of vibrations can stabilize an unstable system \eqref{general} without any measurements of the states \cite{BS-BTZ:97,CX-TY-MI:2018,SMM:80}. 

\subsection{Vibrational Control in Kuramoto-Oscillator Networks}
In the Kuramoto-oscillator network described by \eqref{no_input}, the natural frequencies and coupling strengths can be taken as the parameters. This paper specifically focuses on injecting vibrations solely into the network connections, affecting their coupling strengths (i.e., the strengths of edges in the associated graph). This is inspired by the observation that deep brain stimulation predominantly affects dendrites and axons near the electrode, rather than the soma \cite{HMH-BV-RRC:09}. 

The control inputs exert their influence on the system in the way specified by
\begin{equation}\label{Kuramoto:controlled}
	\dot \theta_i =\omega_i + \sum_{j=1}^{n} \big(w_{ij}+v_{ij}(t)\big) \sin(\theta_j- \theta_i),
\end{equation}
where $v_{ij}(t)$ is the vibration introduced to the edge $(j,i)$. Particularly, we consider that each $v_{ij}(t)$ is simply sinusoidal, which naturally satisfies \eqref{vib:gener_form}, i.e.,
\begin{equation}\label{form:vib}
	v_{ij}(t)= \mu_{ij } \sin(\alpha_{ij} t),
\end{equation}
Let $V(t)=[v_{ij}(t)]_{n \times n}$.
Note that various types of vibrations can be utilized in practical applications. However, for the sake of analysis simplification, we just consider sinusoidal vibrations in this paper.
Here, $\mu_{ij} \in \mathbb{R}$ and $\alpha_{ij}>0$ determine the amplitude and frequency of the vibration injected to the connection $(i,j)$. Note that, since vibrations are usually high-frequency, $\mu_{ij}$ and $\alpha_{ij}$ are often rewritten into the form of $\mu_{ij}=u_{ij}/\varepsilon$ and $\alpha_{ij}=\beta_{ij}/\varepsilon$ to streamline the analysis, where $u_{ij}$ and $\beta_{ij}$ have the order of $1$, and $\varepsilon>0$ is a small constant. Let $u_{ij}(t)=u_{ij}\sin(\beta_{ij} t)$, and $v_{ij}(t)$ can be rewritten as $v_{ij}(t)=\frac{1}{\varepsilon} u_{ij}(\frac{t}{\varepsilon})$.


In contrast to the general system \eqref{general}, where vibrations can be applied to any parameter in $a$, the introduction of vibrational control to a network system is subject to the constraints imposed by the network structure. It is reasonable to assume that vibrations can only be introduced to connections that already exist in the network. Therefore, vibrational control satisfies that for any pair of $i,j$,
\begin{align}\label{constraints}
	&u_{ij}=0, &\text{ if } w_{ij}=0.
\end{align}

\textbf{Objective:} This paper aims to study how vibrational control satisfying both \eqref{form:vib} and \eqref{constraints} stabilizes the cluster synchronization represented by the manifold $\mathcal{M}$.

\section{Preliminary}\label{prelim}
\subsection{Graph-theoretical Notations}\label{graph:notation}
We first introduce some graph-theoretical notations, which are further elucidated in a more intuitive manner in Fig.~\ref{notation}.

For the directed graph $\CG=(\CV,\CE,W)$ that describes the network in \eqref{no_input}, denote the oriented incidence matrix as $B=[b_{k\ell}]\in\R^{n\times m}$, where 
\begin{equation*}
	b_{k\ell}=\begin{cases}
		-1, \text{ if the edge } e_\ell \text{ leave the node } k, \\
		1, \text{ if the edge } e_\ell \text{ enters the node } k, \\
		0, \text{otherwise}.
	\end{cases}
\end{equation*}

For the partition $\CP:=\{\CP_1,\CP_2,\dots,\CP_r\}$ of $\CG$, define $\CG_k=(\CP_k,\CE_k)$ where $\CE_k:=\{(i,j)\in \CE: i,j\in \CP_k\}$ for all $k=1,\dots, r$. Denote $n_k:=|\CP_k|$ as the number of nodes in $\CG_k$. We assume that each $\CG_k$ contains at least 2 nodes and is strongly connected. Let $\CG_\intra=(\CV,\CE_\intra)=\cup_{k=1}^r \CG_k$, and we refer to it as the \textit{intra-cluster subgraph}, describing the intra-cluster network structure. Similarly, let $\CG_{\inter}=(\CV,\CE_\inter)$ be the \textit{inter-cluster subgraph}, where $\CE_\inter :=\CE\backslash \CE_\intra$. Denote $B_\intra$ and $B_\inter$ as the oriented incidence matrices of $\CG_\intra$ and $\CG_\inter$, respectively.

 \begin{figure}[t]
	\centering
	\includegraphics[scale=1.8]{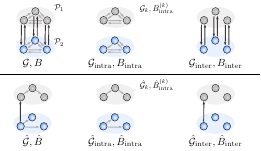}
	\caption{Summary of the main notations we use in this paper: different subgraphs and their corresponding incidence matrices. }
	\label{notation}
\end{figure}

Consider an arbitrary spanning tree in $\CG$, and denote it by $\hat \CG=(\CV,\hat \CE)$, where $|\hat \CE|=n-1$. Let $\hat B$ be the oriented incidence matrix of $\hat \CG$.
For each $k$, Let $\hat \CG_{k}=(\CP_k,\hat \CE_k)$ with  $\hat \CE_k:=\hat \CE\cap \CE_k$. Similar to earlier notations, denote the intra-cluster subgraph of  the spanning tree $\hat \CG$ by $\hat \CG_\intra=(\CV,\hat \CE_\intra):=\hat \cup_{k=1}^r \hat \CG_k$;  the inter-cluster subgraph  of $\hat \CG$ is denoted by $\hat \CG_\inter=(\CV,\hat \CE_\inter)$ with $\hat \CE_\inter:=\hat \CE\backslash \hat \CE_\intra$. Denote the incidence matrices of  $\hat \CG_\intra$ and $\hat \CG_\inter$ as $\hat B_\intra  \in \R^{n\times (n-r)}$ and $\hat B_\inter \in \R^{n \times (r-1)}$, respectively.

\subsection{Incremental Dynamics and Partial Stability}
To study cluster synchronization in \eqref{Kuramoto:controlled}, we look at its incremental dynamics. Specifically, we order the columns of the incidence matrix $\hat B$ of the spanning tree $\hat G$ in a way such that $\hat B=[\hat B_\intra, \hat B_\inter ]$, where $\hat B_\intra\in \R^{n\times n-r}$ and $\hat B_\inter \in \R^{n\times r-1}$ represents the intra- and inter-cluster subnetwork in the spanning tree $\hat G$, respectively. Now, let $x=\hat B_\intra ^\top\theta \in\R^{n-r}$ and $y=\hat B_\inter ^\top \theta \in \R^{r-1}$ to capture intra- and inter-cluster phase differences, respectively. Naturally, $x=0$ implies the cluster synchronization defined by $\CM$ since it indicates that phase differences of oscillators within each cluster are all $0$.
	
Now,  the incremental dynamics of the  Kuramoto-oscillator network can be derived as
\begin{subequations}\label{compact_form}
	\begin{align}
		&\dot x=f_\intra(x) +f_\inter(x,y)+ h_{\ctl}\Big(\frac{1}{\varepsilon}\BmU\big(\frac{t}{\varepsilon}\big),x,y\Big),\label{compact_form:1}\\
		&\dot y=g(x,y)+h'_{\ctl}\Big(\frac{1}{\varepsilon}\BmU\big(\frac{t}{\varepsilon}\big),x,y\Big),
	\end{align}
\end{subequations}
where $\BmU(t)=\diag ([u_{ij}(t)]_{(i,j)\in \CE})\in \R^{m\times m}$ captures the vibrations introduced to the edges in the network. The expressions of the functions on the right-hand side of \eqref{compact_form} are given in Appendix~\ref{derivatin}. We single out some important properties of them. The functions $f_\intra$ and $f_\inter$ results from intra- and inter-cluster couplings; it holds that $f_\intra(0)=0$ and $f_\inter(0,y)=0$ for any $y$. The functions $h_\ctr$ and $h'_\ctr$ are the consequences of vibrational inputs; when there is no control,  $h_{\ctl}(0,x,y)=0$ and $h'_{\ctl}(0,x,y)=0$ for any $x,y$.  

Notice that, the fact that  $f_\intra(x)+f_\inter(x,y)=0$ when $x=0$ signifies the invariance of cluster synchronization defined by $\CM$. To preserve this invariance, a vibrational control input needs to ensure $h_\ctr\big(\frac{1}{\varepsilon}\BmU\big(\frac{t}{\varepsilon}\big),0,y \big)=0$ for any $y$. To meet this requirement, one way is to assume a similar balanced condition as Assumption~\ref{invariance}, that is, one can consider that, given any $\ell$ and $k$, it holds for all $t\ge 0$  that 
\begin{equation}\label{ctr:balance}
	\sum_{q\in \CP_\ell}(u_{iq}(t)-u_{jq}(t))=0 \text{ for any } i,j \in \CP_k.
\end{equation}

This assumption ensures that $x=0$ is an equilibrium of the system \eqref{compact_form:1}, independent of $y$. Following \cite[Chap. 4]{haddad2011nonlinear},  such an equilibrium is called a \textit{partial equilibrium} of the system \eqref{compact_form}.  To ensure the stability of the cluster synchronization manifold $\CM$, one can focus on the stability of the partial equilibrium $x=0$ of the system \eqref{compact_form}, which reduces to studying its partial stability. Note that various forms of partial stability exist, such as partial asymptotic and exponential stability, as elaborated in \cite[Chap. 4]{haddad2011nonlinear}. However, for the scope of our discussion, we exclusively focus on partial exponential stability and therefore present its specific definition here.

\begin{definition}[Partial Exponential Stability {\cite[Chap. 4]{haddad2011nonlinear}}] \label{defi:stability}
	Given a system $\dot x =f(x,y), \dot y= g(x,y)$, where $x\in\R^n, y \in\R^m$,  a partial equilibrium point $x=0$ is \emph{exponentially $x$-stable} uniformly in $y$  if there exist $c_1,c_2,\delta>0$ such that $\|x(0)\|<\delta$ and any $y(0)\in \R^m$ imply that $ \|x(t)\|\le c_1 { \|x_0\|} e^ {-c_2t}$ for all $t\ge 0$. 
\end{definition}

Formally, the manifold $\CM$ is exponentially stable along the system \eqref{Kuramoto:controlled} if $x=0$ of the system \eqref{compact_form2} is partially exponentially stable. To stabilize the cluster synchronization manifold $\CM$, it suffices to design vibration control satisfying \eqref{ctr:balance} to ensure the partial exponential stability of  $x=0$.

\section{Vibrational Stabilization: General Results} \label{results:general}
In this paper, we specifically investigate a form of vibrational control in which vibrations are exclusively introduced to the intra-cluster connections, i.e.,   
\begin{equation}
	u_{ij}(t)=0, 
\end{equation}
for any $i,j$ from different clusters. Then, the requirement \eqref{ctr:balance} is naturally met since  $\sum_{q\in \CP_\ell}u_{iq}(t)=0$ for any $i$. As a result, the system \eqref{compact_form} reduces to (see details in Appendix~\ref{derivatin})
\begin{subequations}\label{compact_form2}
	\begin{align}
		&\dot x=f_\intra(x) + f_\inter(x,y)+ \frac{1}{\varepsilon} h_{\ctl}\Big(\BmU \big (\frac{t}{\varepsilon}\big),x\Big),\label{compact_form2:1}\\
		&\dot y=g(x,y)+\frac{1}{\varepsilon}h'_{\ctl}\Big(\BmU \big (\frac{t}{\varepsilon}\big),x\Big),
	\end{align} 
\end{subequations}
where $h_\ctr$ and $h'_\ctr$ no longer depend on $y$.

Since $f_\inter(0,y)=0$ holds for any $y$, then the term $f_\inter(x,y)$ in \eqref{compact_form2:1} can be viewed as a vanishing perturbation dependent of $y$ to the controlled nominal system
\begin{align}\label{nominal}
	\dot x=f_\intra(x)+\EP f_\ctr\Big( \BmU\big(\EPt\big),x\Big).
\end{align}
This perturbation can be decomposed as 
\begin{align*}
	f_\inter=[(f^{(1)}_{\inter})^\top,\dots,(f^{(r)}_{\inter})^\top]^\top
\end{align*}
where $f^{(k)}_{\inter}=- (\hat B^{(k)}_{\intra})^\top \BB_\inter \BmW_\inter \sin(R_2 x+R_3y)$ is the perturbation received by the $k$th cluster\footnote{The matrix $\BB$ is obtained by replacing the negative elements in the oriented incidence matrix $B$ with $0$. The columns of $\BB$ can be ordered such that $\BB=[\BB_\intra,\BB_\inter]$ and $\BB_\intra=[\BB_\intra\PC{1},\BB_\intra\PC{2},\dots,\BB_\intra\PC{r}]$. The diagonal matrix $\BmW=\blkdiag(\BmW_\intra,\BmW_\inter):=\diag([w_{ij}]_{(i,i)\in \CE})$ contains the weights, and the matrices $R_1$, $R_2$, and $R_3$ capture the relation between $\theta$ and $x,y$. More details can be found in Appendix~\ref{derivatin}.}.

Next, we show how a vibration control can stabilize $x=0$ in the presence of the perturbation $f_\inter(x,y)$ by introducing the change $\frac{1}{\varepsilon}f_\ctr(\BmU,x)$ to the system dynamics. 

To this end, we linearize the system at $x=0$ and obtain
\begin{align}\label{linearized}
	\scalemath{1}{\dot {\bar x} = \Big(J + \EP P\big(\EPt \big) \Big){\bar x}+ N(y){\bar x}},
\end{align}
where
\begin{equation}\label{expres:JandP}
	\begin{aligned}
		&J=\frac{\partial f_\intra}{\partial x}(0)= \blkdiag (J\PC{1},\dots,J\PC{r}),  \\
		& P(t)=\frac{\partial f_\ctr}{\partial x}(0,t)=  \blkdiag (P\PC{1}(t),\dots,P\PC{r}(t)),\\
		&N(y) = \frac{\partial f_\inter}{\partial x}(0,y).
	\end{aligned} 
\end{equation}
Notice that both $J$ and $P(t)$ are block-diagonal, because the dynamics described by $f_\intra$ and the vibrational control $f_\ctr$ have no inter-cluster couplings. Here,  for each $k$,  ${J\PC{k}=-(\hat B^{(k)}_{\intra})^\top \BB_\intra^{(k)} \BmW^{(k)}_\intra R_1}$ and $P\PC{k}(t)=-	(\hat B_\intra^{(k)})^\top \BB^{(k)}_\intra \BmU^{(k)}(t) R_1$. One can also derive that
\begin{align*}
	N(y) = -(\hat B_\intra)^\top \BB_\inter \BmW_\inter (\BONE_{n-r}^\top \otimes\sin(R_3y))\odot R_2.
\end{align*}
Observe that $P(t)$ is almost periodic with a zero mean value just as $\BmU(t)$. Different from our previous work \cite{YQ-DSB-FP:22a}, here we also linearize the perturbation $f_\inter(x,y)$ at $x=0$ in \eqref{linearized}. 

\begin{lemma}[\textit{Connecting the stability of Systems} \eqref{compact_form2} \textit{and} \eqref{linearized}] \label{lemma:linearized}
	If the equilibrium $\bar x=0$ of the system \eqref{linearized} is exponentially stable uniformly in $y$, then $x=0$ of the system  \eqref{compact_form2} is also exponentially stable uniformly in $y$. 
\end{lemma}

From this lemma, a vibrational control resulting in $P(t)$ that stabilizes $x=0$ of the system \eqref{linearized} stabilizes the cluster synchronization manifold $\CM$, too. Next, we aim to provide conditions on which a vibrational control stabilizes the system \eqref{linearized}. 

Let $s={t}/{\varepsilon}$, then we rewrite the system \eqref{linearized} as
\begin{align}\label{change:time:scale}
	\frac{d \bar x}{d s} = (\varepsilon J +  P(s))\bar x + \varepsilon N(y)\bar x.
\end{align}

Next, we use averaging methods to analyze this system. However, the standard first-order averaging\footnote{Given a system $\dot x = \varepsilon f(t,x) + \varepsilon^2 h(t,x)+\dots+\varepsilon^k g(t,x)$, the first averaging method calculates the averaged system $\dot x = \varepsilon \bar f(t,x)$ only using the first-order term $\varepsilon f(t,x)$ and ignoring the higher-order terms $\mathcal O(\varepsilon^2)$, i.e., $\bar f(x) =\lim_{T\to \infty}\frac{1}{T}\int_{t=0}^{T} f(t,x)dt$ \cite{SJA-FMJ:07}.} is not applicable here. Recall that $P(s)$ has zero mean. Then, applying the first-order averaging to \eqref{change:time:scale} just eliminates the $P(s)$ term and results in the uncontrolled system $\frac{d \bar x}{d s} = \varepsilon J \bar x+\varepsilon N( y) \bar x$. 

To avoid that, we change the coordinates of  \eqref{change:time:scale} first before using the averaging method. To do that, we introduce an auxiliary system  
\begin{align}\label{periodic}
	{\frac{d \hat x}{d s}=P(s) \hat x},
\end{align}
and let $\Phi(s,s_0)$ be its state transition matrix. Since $P$ is block-diagonal, it holds that
$
\Phi=\blkdiag(\Phi\PC{1},\dots,\Phi\PC{r}),
$
where $\Phi\PC{k}$ is the transition matrix of the subsystem in the $k$th cluster ${d \hat x_k}/{d s}=P\PC{k}(s) \hat x_k$. 

Consider the change of coordinates $z(s)=\Phi^{-1}(s,s_0) \bar x(s)$. It follows from the system \eqref{change:time:scale} that
\begin{align}\label{coordinated}
	\frac{dz}{ds} = \varepsilon \Phi^{-1} J \Phi z+ \varepsilon \Phi^{-1} N(y) \Phi z.
\end{align}
Since $P(s)$ is almost periodic in $s$ and mean-zero, so are $\Phi(s,s_0)$ and $\Phi^{-1}(s,s_0)$. Denote
\begin{equation*}
	G(y) := \Phi^{-1} N(y) \Phi.
\end{equation*}
For each $k=1,\dots, r$, let $G^{(k)}(y)\in \R^{(n_i-1)\times n}$ be the dynamics associated with the $k$th cluster.
Now, we associate \eqref{coordinated} with a partially averaged system
\begin{align}\label{average}
	{\frac{dz}{ds} = \varepsilon \big(\bar J + G(y) \big)z},
\end{align}
where
\begin{align}\label{J-bar}
	{\bar J= \lim\limits_{T\to \infty}\frac{1}{T}\int_{s_0}^{s_0+T} \Phi^{-1} (s,s_0) J \Phi(s,s_0)ds}.
\end{align}

As both $J$ and $\Phi$ are block-diagonal, one can derive that $\bar J$ is also block-diagonal satisfying $\bar J=\blkdiag(\bar J\PC{1},\dots,\bar J\PC{r})$ with 
\begin{align}\label{bar_J:k}
	\bar J \PC{k}= \lim\limits_{T\to \infty} \frac{1}{T}\int_{s_0}^{s_0+T} \Big(\Phi\PC{k}(s,s_0)\Big)^{-1} J\PC{k} \Phi\PC{k}(s,s_0)ds.
\end{align} 
Therefore, it holds that 
\begin{equation*}
	\frac{d z^{(k)}}{ds} = \varepsilon \big(\bar J\PC{k} z_k+ G\PC{k}(y) z \big),
\end{equation*}
Then, it can be shown that for any $k$, there exist $\bar \gamma_{k\ell}>0, \ell =1 ,\dots, r$, such that 
\begin{align*}
	\left\|G\PC{k}(y) z  \right\|\le \sum_{\ell=1}^{r}\bar \gamma_{k\ell} \|z_\ell\|
\end{align*}
for each $k$ (see Lemma~\ref{lemma:bound:pert} in Appendix~\ref{proof:general} for more details). 

\begin{theorem}[{Sufficient condition for vibrational stabilization}]\label{theorem:general}
	Assume that $\bar J=\blkdiag(\bar J\PC{1},\dots,\bar J\PC{r})$ in Eq.~\eqref{J-bar} is Hurwitz. Let $\bar X_k$ be the solution to the Lyapunov equation 
	\begin{align}\label{Ly:controlled}
		{(\bar J\PC{k})^\top \bar X_k+\bar X_k \bar J\PC{k}=-I}.
	\end{align}
	Define the matrix $S=[s_{k \ell}]_{r\times r} $ with
	\begin{align}\label{matrix:S}
		s_{k\ell}=\Big\{\begin{matrix*}[l]
			\lambda_{\max}^{-1}(\bar X_k)-\bar\gamma_{k k} , &\text{ if }k=\ell,\\
			-\bar \gamma_{k \ell}, &\text{ if }k \neq\ell,
		\end{matrix*}
	\end{align}
	where $\lambda_{\max}(\cdot)$ denotes the maximum eigenvalue of a matrix. 
	If $S$ is an $M$-matrix\footnote{A real non-singular matrix $A=[a_{ij}]$ is an $M$-matrix if $a_{ij}\le 0$ for any $i\neq j$, and all its leading principal minors are positive \cite[ch. 2.5]{RAH-CRJ:94}.}, then  there exists $\varepsilon^*>0$ such that, for any $\varepsilon<\varepsilon^*$:
	
	(i) the equilibrium $x=0$ of the system \eqref{compact_form2} is exponentially stable uniformly in $y$; 
	
	(ii) the cluster synchronization manifold $\CM$ of the system \eqref{Kuramoto:controlled} is exponentially stable. 
\end{theorem}

Note that a similar theorem was presented in our preliminary work \cite{YQ-DSB-FP:22a}. Here, our results are built on a complete instead of a partial linearization technique in \eqref{linearized}. Theorem~\ref{theorem:general} provides a sufficient condition for vibrational control inputs to stabilize the cluster synchronization. To design an effective vibrational control law that stabilizes $\CM$, one can design vibrations to satisfy the following three conditions: 
i) $\bar J$ in \eqref{J-bar} is Hurwitz, 
ii) $S$ defined in \eqref{matrix:S} is an $M$-matrix, and
iii) the frequency of the vibrations is sufficiently high.

\noindent
\textbf{Connections with robustness of linear systems:} Consider a stable linear system 
\begin{equation*}
	\dot x =Ax.
\end{equation*}
Some earlier works (e.g., \cite{PRV-TM:80,YR:85}) utilize 
\begin{equation}\label{measure:robust}
	\CR(A):=\lambda^{-1}_{\max}(X)
\end{equation}
to measure its robustness ($\lambda_{\max}(\cdot)$ stands for the maximum eigenvalue of a matrix),  where $X$ is the solution to the Lyapunov equation:
\begin{equation*}
	A^\top X +XA =-I. 
\end{equation*}
A larger $\CR(A)$ means that the system is more robust.

In our case, from \eqref{linearized}, the uncontrolled intra-dynamics around $\CM$ are described by the linearized system
\begin{align}\label{uncontrolled}
	&\dot x_k=J\PC{k} x_k+f\PC{k}_\inter(x,y),&k=1,\dots,r,
\end{align}
where $J\PC{k}$ is stable and $f\PC{k}_\inter(x,y)$ is taken as the vanishing perturbation. Here, $x_k=0$ means synchronization of the oscillators in the $k$th cluster. Similarly, one can interpret that $\CR(J\PC{k})$ measures the robustness of synchronization in the $k$th cluster.  If the intra-cluster synchronization is sufficiently robust (i.e., $\CR(J\PC{k})$'s are sufficiently large) to dominate the perturbations resulted from inter-cluster connections, the cluster synchronization is stable. A sufficient condition is constructed in \cite[Th. 3.2]{TM-GB-DSB-FP:18}. By contrast, if $\CR(J\PC{k})$'s are not large enough, the cluster synchronization can lose its stability. Yet, the robustness of the intra-cluster synchronization can be reshaped by introducing vibrations to the local network connections. The new robustness is instead measured by $\CR(\bar J\PC{k})$ defined in \eqref{bar_J:k}.   

Now, the question naturally arises: how to design vibrational control such that the robustness in the cluster can be improved? We aim to provide answers in the next section.

\section{Improving Robustness by Vibrational Control}\label{sec:robustness}

The primary objective of this section is to demonstrate the design of vibrational control with the intention of enhancing the robustness of synchronization within each cluster. As observed in the concluding part of the previous section, the robustness is intimately linked to the linearized systems of both the uncontrolled system and the averaged controlled system. Hence, we commence by examining the linear system and subsequently explore the applicability of the findings from linear systems to Kuramoto-oscillator networks.

\subsection{Linear Systems}\label{subsec:linear}
Let us consider a linear system 
\begin{equation}\label{linear}
	\dot x =A x,
\end{equation}
where $x\in\R^n, A\in\R^{n\times n}$, and $A$ is assumed to be Hurwitz. Consider a vibrational control matrix $U(t)$ that influences the system parameters in $A$, resulting in the following controlled system
\begin{equation}\label{controlled_net_compact}
	\dot x = \big(A+\frac{1}{\varepsilon}U\big(\frac{t}{\varepsilon}\big) \big) x,
\end{equation} 
where  $U(t)=[u_{ij}(t)]_{n\times n}$ with $u_{ij}(t)=u_{ij}\sin(\beta_{ij} t)$, and $\varepsilon>0$ is small, determining the frequencies of the vibrations.


Vibrational control can improve the robustness of a stable system. To show this, we follow similar step as in Section~\ref{results:general} to associate \eqref{controlled_net_compact} with the averaged system
\begin{equation}\label{linear:averaged}
	\dot{\bar x} = \bar A \bar x, 
\end{equation}
where the time scale has been restored to $t=\varepsilon s$, 
\begin{equation*}
	\bar A =\lim_{T\to \infty}\frac{1}{T} \int_{t=0}^{T} \Phi^{-1}(t,t_0) A \Phi(t,t_0) d t,
\end{equation*}
and $\Phi(t,t_0)$ is the state transition matrix of the system
\begin{equation}\label{auxilary}
	\dot {\hat x} = U(t) \hat x. 
\end{equation}

When $\bar A$ is Hurwitz, the controlled system \eqref{controlled_net_compact} behaves like \eqref{linear:averaged} in average if $\varepsilon$ is small. Then, one can interpret that vibrational control changes the system matrix from $A$ to $\bar A$ in an in-average sense. Vibrational inputs can be designed to carefully modify the elements in $A$ so that $\CR(\bar A)$ is larger than $\CR(A)$ to improve robustness. 

However, which elements in $A$ and how they can be changed is a challenging problem. An earlier attempt has been made in \cite{SMM:80}. Here, we aim to generalize their result by utilizing graph-theoretical approaches. 

Specifically, we associate the uncontrolled system $\dot x =Ax$ with a weighted directed network $\CG_A=(\CV,\CE_A,A)$. Here, $\CV=\{1,2,\dots\}$, and there is a directed edge from $i$ to $j$, i.e., $(i,j)\in\CE_A$ if $a_{ji}\neq 0$ and $i\neq j$ (no self-loops are considered). The matrix $A$  becomes the weighted adjacency matrix. Likewise, one can associate the averaged controlled system \eqref{linear:averaged} with a weighted directed network $\CG_{\bar A}=(\CV,\CE_{\bar A},{\bar A})$, which we refer to as the \textit{functioning network}\footnote{We highlight that the terminology we employ here is different from \textit{functional network} that is widely-used in neuroscience (e.g., see \cite{DBS-MLA-ASD-TBE:06}). To avoid possible confusion,  we clarify that, when we say that vibrational control functionally changes the network $\CG_A$, we mean that it leads to a functioning network $\CG_{\bar A} $ that is different from $\CG_A$.}. Then, changing elements in $A$ reduces to altering the connection weights in $\CG_A$ via vibrational control. 

\begin{definition}
	The edge $(i,j)\in \CE_A$ is said to be vibrationally increasable if there exists a vibrational control $U(t)$ such that the weight of $(i,j)$ is increased in $\CE_{\bar A}$, i.e.,  $\bar a_{ji}>a_{ji}$. It is said to be vibrationally decreasable if there exists a vibrational control $U(t)$ such that the weight of $(i,j)$ is decreased in $\CE_{\bar A}$, i.e.,  $\bar a_{ji}<a_{ji}$.
\end{definition}

\begin{lemma}\label{lemma:inc_dec}
	Consider an edge $(i,j)\in \CE_A$. It is vibrationally increasable if there is an edge in the reverse direction that has a negative weight, i.e., $a_{ij}<0$. It is vibrationally decreasable if there is an edge in the reverse direction that has a positive weight, i.e., $a_{ij}>0$. 
\end{lemma}

 \begin{figure}[t]
	\centering
	\includegraphics[scale=1]{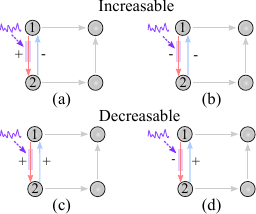}
	\caption{Illustration of vibrationally increasable and decreasable edges. (a) and (b): The red edge (i.e., $(1,2)$) is vibraionally increasable since the edge in the reverse direction has negative weight. (c) and (d): It is vibrational decreasable since the reverse edge has positive weight. Injecting a vibration to the red edge itself can functionally increase or decrease its weight in the corresponding cases.}
	\label{illus:decr_incr}
\end{figure}

An illustration of vibrationally increasable and decreasable edges can be found in Fig.~\ref{illus:decr_incr}. When an edge satisfies the corresponding conditions, we further find that directly injecting a vibration to it can functionally increase or decrease its weight (see Appendix~\ref{pf:inc_dec} for more details).

The conditions we identified here are just sufficient ones. Edges that do not satisfy these conditions may  also be vibrationally increasable or decreasable. 
Yet, we restrict our attention to the edges that satisfy the conditions in Lemma~\ref{lemma:inc_dec}. Then, based on them, we define two sets $\CE_{\rm inc}:=\{(i,j)\in \CE_A: a_{ij}<0\}$, and $\CE_{\rm dec}:=\{(i,j)\in \CE_A: a_{ij}>0\}$, which are vibrationally increasable and decreasable edges, respectively. Subsequently, we define a directed and signed graph $\CG^{\rm mod}_A:=(\CV,\CE_{\rm mod},S)$, where $\CE_{\rm mod}=\CE_{\rm inc} \cup  \CE_{\rm dec}$, and $S=[s_{ij}]$ with $s_{ij}=-\sign(a_{ij})$ if $(i,j)\in \CE_A$ and $s_{ij}=0$ otherwise. We refer to $\CG^{\rm mod}_A$ as the \textit{modifiable graph} of $\CG_A$. An example is shown in Fig.~\ref{theo:linear} (b). Given a graph $\CG_\Delta$, we denote $\CG_\Delta\subseteq \CG^{\rm mod}_A$ if the signed (but unweighted) edges of $\CG_{\Delta}$ constitute a subset of $\CG^{\rm mod}_A$'s. 

\begin{theorem}\label{vibrational:design}
	Consider a matrix $\Delta=[d_{ij}]\in\R^{n\times n}$, and let $\CG_\Delta:=(\CV,\CE_\Delta,\Delta)$ be the directed and signed graph associated with it. If $\CG_\Delta$ is a directed acyclic graph (i.e., a graph with no cycles) and $\CG_\Delta \subseteq \CG_A^{\rm mod}$, there exist vibrational control inputs such that the system matrix $\bar A$ of \eqref{linear:averaged} becomes $\bar A=A+\Delta$. Further, if  $A+\Delta$ is Hurwitz., there exist $\varepsilon_0$ such that, for any $\varepsilon<\varepsilon_0$, the system \eqref{controlled_net_compact} is exponentially stable.

\end{theorem}

This theorem provides a method to functionally change the system matrix from $A$ to $A+\Delta$, which stabilizes the system if $A+\Delta$ is Hurwitz. A matrix $\Delta$, which is associated with an acyclic graph $\CG_\Delta$ and $\CG_\Delta\subseteq \CG_A^{\rm mod}$, is called as \textit{realizable modification matrix}, and $\CG_\Delta$ is called \textit{realizable modification graph}, given that,  under these conditions, there exist vibrational inputs to realize the desired functional changes. Next, we show how to design such control inputs. 

By assumption, $\CG_\Delta$ is directed acyclic, then, from \cite{JBJ-GG:00},  it can be topologically ordered. Therefore, there exists a permutation matrix $Q$ such that the matrix $\Delta': =Q\Delta Q^{-1}$ is quasi-lower-triangular. One can let $A'=QAQ^{-1}$, and $U'(t)$ be the vibrational control matrix to $A'$. Let $x'=Qx$, one can derive that the controlled system becomes
\begin{align}\label{linear:system:orig}
	\dot x' = \Big(A'+\frac{1}{\varepsilon}U'\big(\frac{t}{\varepsilon}\big) \Big)x'. 
\end{align}
Now, to functionally change $A$ to $A+\Delta$, it becomes to change $A'$ to $A'+ \Delta'$.

To realize this functional change, we consider the following vibrational control matrix that is quasi-lower-triangular:
\begin{equation*}
	U'(t)=\begin{bmatrix}
		0&0& 0&\dots &0\\
		u'_{21}(t)&0 & 0&\dots &0\\
		u'_{31}(t) & u'_{32}(t) & 0 &\dots &0\\
		\vdots& \vdots& \vdots & \ddots &\vdots\\
		u'_{n1}(t) & u'_{n2}(t) & u'_{n3}(t) &\dots &0
	\end{bmatrix},
\end{equation*}
where $u'_{ij}(t)=u'_{ij}\sin(\beta'_{ij} t)$. Then, the state transition matrix of the system $\dot{\hat x}=U(t) \hat x$ has the following form:
\begin{equation*}
	\Phi'(t,t_0)=\begin{bmatrix}
		1&0& 0&\dots &0\\
		\phi'_{21}(t)&1 & 0&\dots &0\\
		\phi'_{31}(t) & \phi'_{32}(t) & 1 &\dots &0\\
		\vdots& \vdots& \vdots & \ddots &\vdots\\
		\phi'_{n1}(t) & \phi'_{n2}(t) & \phi'_{n3}(t) &\dots &1
	\end{bmatrix}.
\end{equation*}
According to \cite{SMM:80}, the averaged system matrix of \eqref{linear:system:orig}, denoted as $\dot{\bar x}=\bar A' \bar x $, satisfies $\bar A'=A'+\bar B$ with
\begin{align}\label{B_bar}
	\bar B =[\bar b_{ij}]= A^\top \odot C. 
\end{align}
Here, $C=[c_{ij}]$ satisfies
\begin{equation}\label{c_ij}
	c_{ij}= -\frac{1}{T} \int_{t=0}^{T}\big( \phi_{ij}' (\tau)\big)^2 dt .
\end{equation}
We need to design vibrational control inputs such that $\bar B = \Delta$. In other words, one needs to design vibrations such that
\begin{equation}\label{desin:c_ij}
	c'_{ij}= -\frac{d'_{ij}}{a'_{ji}}, \text{ if } a'_{ji}\neq 0; c'_{ij}= 0, \text{ otherwise}.
	\end{equation}
One can derive that, for any $i \ge 2$, 
\begin{align}
	\phi'_{i.i-1} (t)&=\int_{t_0}^{t} u'_{i,i-1}(\tau)  d \tau \nonumber\\
	& = -\frac{u'_{i,i-1}}{\beta'_{i,i-1} } \big(\cos(\beta'_{i,i-1}t)-\cos(\beta'_{i,i-1}t_0) \big), \label{phi:1}\\
	\phi'_{ij} (t) &=\int_{t_0}^{t} \sum_{k=1}^{i-1} u'_{ik}(\tau) \phi'_{kj} (\tau) d \tau, \text{ for } j \le i-2.\label{phi:2}
\end{align}

Combining the expressions in \eqref{phi:1} and \eqref{phi:2} with \eqref{c_ij}, one can derive the required ratios of the amplitudes and frequencies, i.e., $\frac{u'_{ij}}{\beta'_{ij}}$, in $U'(t)$.  First, for each $i\ge 2$, it follows from  \eqref{phi:1} and \eqref{c_ij} that
\begin{equation}\label{left_of_diagonal}
	\frac{u'_{i,i-1} }{\beta'_{i,i-1}}=\sqrt{2c'_{i,i-1}}=\sqrt{\frac{-2d'_{i,i-1}}{a'_{i-1,i}}}, \text{ if } a'_{i-1,i}\neq 0.
\end{equation}
Next, one can establish the remaining parameters within the vibrational control matrix through a recursive method. To do this, one calculates the necessary ratio $\frac{u'_{ij}}{\beta'_{ij}}$ by systematically deriving each required $\phi'_{ij}$ in the sequence outlined in Fig.~\ref{vib_design}.
\begin{figure}[t]
	\begin{center}
		\includegraphics[scale=0.7]{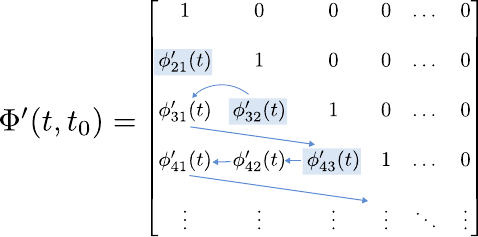}
	\end{center}
	\caption{Vibtaional control design. We first design the ratio $u'_{ij}/\beta'_{ij}$ in each vibrational  input $u'_{ij}(t)$, which can be done by  once $\phi'_{ij}(t)$ is derived. The entries that are one-element left to the diagonals of $\Phi'$ are expressed by \eqref{phi:1}, which leads to the corresponding ratios in \eqref{left_of_diagonal}.  The remaining entries of $\Phi'$ can be determined one by one in the order depicted by the above arrows, which can also be used to determine the ratios.  After all the ratios are determined, one can choose amplitudes and frequencies to satisfy the corresponding ratio.}
	\label{vib_design}
\end{figure} 

Once the ratios $\frac{u'_{ij}}{\beta'_{ij}}$ have  been determined, one can select amplitudes and frequencies to fulfill these ratios. Note that the frequencies $\beta'_{ij}$ need to be incommensurable, that is, each $\beta'_{ij}/\beta'_{k\ell}$ for any distinct pair of $\{i,j\}$ and $\{k,\ell\}$ is not a rational number. Consequently, the vibrational control matrix $U('t)$ is determined. Subsequently, as $A=Q^{-1}A' Q$,  the vibrational control to the original system \eqref{controlled_net_compact} becomes 
\begin{equation*}
	U(t)=Q^{-1}U'(t)Q.
\end{equation*}

\begin{figure*}[t]
	\centering
	\includegraphics[scale=1]{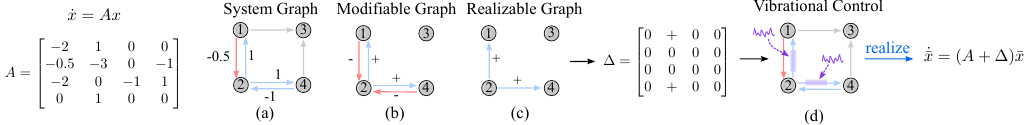}
	\caption{Illustration of the method to improve robustness of a linear stable system. (a) The directed graph associated with $\dot x =A x$. (b) Modifiable graph, where edges can be vibrationally changed (the signs indicate whether they can be increased or decreased). (c)-(d) If a matrix $\Delta$ corresponds to a singed and directed graph that is  directed acyclic and it is a subgraph of the modifiable graph, then there exists vibrational control to realize the averaged system $\dot {\bar x} =(A+\Delta)\bar x$. Careful design of $\Delta$ can increase the robustness of the original system.}
	\label{theo:linear}
\end{figure*} 

\begin{figure*}[t]
	\centering
	\includegraphics[scale=1]{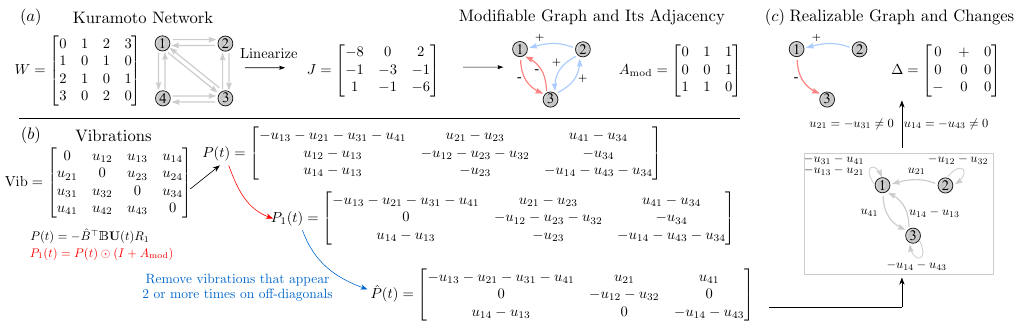}
	\caption{Illustration of making predictable changes in the Kuramoto-oscillator network. (a) One linearizes the Kuramoto model and obtain a linear system $\dot x = J x$. Following the steps in Section~\ref{sec:robustness}-\ref{subsec:linear}, a modifiable graph and its adjacency matrix can be found. (b) We investigate how vibrations in the Kuramoto-oscillator network influence the linearized system by computing $P=\hat B^\top \BB \BmU(t)R_1$ (in the figure, the incidence matrix $\hat B$ corresponds to the spanning tree  consisting of the edges $\{(3,1),(3,2),(3,4)\}$). From the modifiable graph in (a), we only keep elements in $P(t)$ that corresponds to modifiable edges, resulting in $P_1(t)$. Further, to make the design of vibrational control tractable, we remove the vibrations that appear two or more times in off-diagonal positions and obtain $\hat P(t)$. (c) We associate $\hat P(t)$ with a directed graph $\CG_{\hat P}$ that allows for self-loops. Then, it remains to configure vibrations such that the directed graph has no directed cycles (including self-loops). For instance, one can set $u_{21}(t)=-u_{31}(t)$, $u_{14}(t)=u_{43}(t)$, and any other vibrations to zero to realize the changes indicated in the upper panel. Any change that has the same pattern as $\Delta$ can be realized by choosing the amplitudes and frequencies in the vibrations $u_{21}(t)$ and $u_{14}(t)$ following the steps in Section~\ref{sec:robustness}-\ref{subsec:linear}.}
	\label{design:kuramoto}
\end{figure*}

\begin{remark}
Theorem~\ref{vibrational:design} also provides an approach to improve the robustness of the system by vibrational control. As illustrated in Fig.~\ref{theo:linear}, if there is a matrix $\Delta$ satisfying the conditions (i)-(ii) and $\CR(A+\Delta)>\CR(A)$, one can follow the above steps to design vibtational inputs to realize this functional change.
We wish to mention that it is likely impossible to improve the system to any desired robustness by just adding a matrix $\Delta$, especially when $\Delta$ is constrained by the graph structure. There are some interesting open questions. For instance, what is the realizable range of robustness levels? How to design $\Delta$ and the subsequent vibrational control to realize a desired and reasonable robustness? 
\end{remark}

\subsection{Kuramoto-Oscillator Networks}\label{subsec:kuramo}
Observe that, in each cluster, oscillators have an identical frequency, and each pair is coupled by bidirected edges with asymmetric strengths. To study how vibrational control can improve the robustness of synchronization in each cluster, we consider the following homogeneous Kuramoto model:
\begin{equation}
	\dot \varphi_i = \omega + \sum_{j=1}^{n} w_{ij} \sin(\varphi_j -\varphi_i),
\end{equation}
where $i=1,\dots,n$, and $w_{ij}$'s describe the directed network $\BG=(\BV,\BE)$ with $|\BE|=m$. Let $B\in \R^{n \times m}$ be the incidence matrix of $\BG$.  Select a directed spanning tree $\BG_{\rm span}$ in $\BG$, and let $\hat B\in \R^{n\times (n-1)}$ be its incidence matrix. Denote $\varphi=[\varphi_1,\dots,\varphi_m]$ and $\BmW=\diag([w_{ij}]_{(i,j)\in \BE})\in\R^{m\times m}$.

Let $x= \hat B^\top \varphi\in \R^{n-1}$, and following similar steps as Appendix~\ref{derivatin} one can derive that
\begin{equation}\label{homog:compact}
	\dot x = - \hat B^\top \BB \BmW \sin (R_1 x),
\end{equation}
where $\BB$ is obtained by replacing the negative elements in the oriented incidence matrix $B$ with $0$ and
\begin{equation*}
	R_1 = \begin{bmatrix}
		B^\top (\hat B^\top)^\dagger\\
		-B^\top (\hat B^\top)^\dagger
	\end{bmatrix}.
\end{equation*}

With vibrations injected into the edges in the network, we have the controlled model
\begin{equation}\label{homog:compact:controlled}
	\dot x = - \hat B^\top \BB \Big(\BmW+\frac{1}{\varepsilon}\BmU \big(\frac{t}{\varepsilon}\big)\Big) \sin (R_1 x),
\end{equation}
where $\BmU(t)=\diag([u_{ij}]_{(j,i)\in\BE})\R^{m\times m}$. 
Linearizing the system \eqref{homog:compact:controlled} at $x=0$, we obtain
\begin{equation}\label{ctr_KM:linearized}
	\dot x = - \hat B^\top \BB \Big(\BmW+\frac{1}{\varepsilon}\BmU\big(\frac{t}{\varepsilon}\big)\Big) R_1 x. 
\end{equation}
Denote $J=-\hat B^\top \BB \BmW R_1$. Similar to the previous section, one can associate \eqref{ctr_KM:linearized} with the following averaged system
\begin{equation}
	\dot{\bar x} = \bar J \bar x
\end{equation}
where 
\begin{equation*}
	\bar J =\lim_{T\to \infty}\frac{1}{T} \int_{t=0}^{T} \Phi^{-1}(t,t_0)J\Phi(t,t_0) d t,
\end{equation*}
and $\Phi(t,t_0)$ is the state transition matrix of the system
\begin{equation}
	\dot z = P(t) z, \text{where } P(t) =  - \hat B^\top \BB \BmU(t)R_1.
\end{equation}

Due to the presence of the matrices $\hat B,\BB$, and $R_1$, $P(t)$ has a very complex dependence on the vibrations injected to the edges in the Kuramoto-oscillator network. As a consequence, it becomes very challenging to design vibrational control $U(t)$ to modify the elements in $J$ so that the robustness of synchronization can be improved.

Our goal is to provide a tractable and predictable approach to design vibrational control to improve the robustness of the synchronization by functionally modifying the elements in $J$.
To this end, we will use the results we established for linear systems in Section~\ref{sec:robustness}-\ref{subsec:linear}.

Specifically, the $\CG_\Delta$ consists of the following steps (an example is provided in Fig.~\ref{design:kuramoto} to illustrate the procedure). 

(1) First, selecting a $\hat B$ we compute the Jacobian matrix $J$ and associate it with a weighted directed graph $\BG_J$. 

(2) Following the same steps in Section~\ref{sec:robustness}-\ref{subsec:linear}, one can identify a modifiable graph from $\BG_J$, denoted as $\BG_J^{\rm mod}$, containing  increasable and decreasable edges in $\CG_J$. Here, an edge $(k,\ell)$ is increasable (decreasable) if $j_{k\ell}<0$ (resp., $j_{k\ell}>0$). The definition here differs slightly from the one presented in the previous section. In linear network systems, an edge that does not exist is not considered part of the set of modifiable edges, as it cannot receive control inputs. However, in the linearized model of a Kuramoto network, a zero entry in the Jacobian matrix can still be influenced by vibrational control injected into the connections of the original Kuramoto network (an example is shown in Fig.~\ref{design:kuramoto}~(a)). Let $A_{\rm mod}$ be the unweighted adjacency matrix of $\CG_J$

(3) We assume that $\BmU(t)$ contains non-zero values at all off-diagonal positions (we will gradually set them to $0$ at positions that we do not intend to control). To assess how vibrations injected into the edges of the Kuramoto-oscillator network impact the edges within $\BG_J$, we calculate $P(t) = -\hat B^\top \BB \BmU(t)R_1$. Since whether an edge in $\BG_J$ can be altered is determined by its modifiable graph, we let $P_1(t) = P(t)\odot (I+A_{\rm mod})$. This operation retains only the elements corresponding to modifiable edges.

(4) While configuring control inputs, one needs to deal with the situation that a vibration introduced to a single edge in the Kuramoto-oscillator network can bring changes to multiple edges in $\BG_J$.  To make the design procedure more analytically tractable, we remove the vibrations that appear two or more times in the off-diagonal positions in $P_1(t)$ and obtain $\hat P(t)$. 

(5) We associate $\hat P(t)$ with a directed graph $\BG_{\hat P}$. Now, we configure vibrational control inputs such that $\BG_{\hat P}$ does not contain a directed cycle (also no self-loops). Consequently, the resulting graph determines realizable changes to $J$ that vibrations can bring in, which we refer to as a \textit{realizable graph}. For any $\Delta$, there always exists a vibrational control that functionally changes $J$ to $\bar J =J+\Delta$ if the associated directed and sign graph of $\Delta$ is a realizable graph (see Fig.~\ref{design:kuramoto}~(c)). One can simply use the results in Section~\ref{sec:robustness}-\ref{subsec:linear} to design vibrational inputs.

\begin{remark}
	The vibrations that exclusively affect the diagonal positions of $P(t)$ play a crucial role in the design process. They are often utilized to counteract the impact of other vibrations on the diagonal, ensuring that they only influence a single off-diagonal element in $P(t)$. To maximize the occurrence of vibrations exclusively in the diagonal positions, an effective approach is to select a spanning tree with minimal depth to define $\hat B$. For example, in Fig.~\ref{design:kuramoto} (a), one can opt for the spanning tree composed of edges ${(3,1),(3,2),(3,4)}$. This choice facilitates the isolation of vibrations to the diagonal positions.
	\end{remark}

\subsection{Design of Vibrational Control for Cluster Synchronization Stabilization}
Now, we can use the results in the previous section to design vibrations to stabilize cluster synchronization. 

Recall that oscillators in each cluster have an identical frequency. Following the same procedure as in the previous section, one can identify a modifiable graph for each cluster, defining the edges that can be functionally modified via vibrational control. Denote the modifiable graphs by $\BG^{(1)}_{\rm mod}, \BG^{(2)}_{\rm mod},, \dots,\BG^{(r)}_{\rm mod}$.

\begin{corollary}\label{coro}
	Consider matrices $\Delta^{(k)}\in \R^{n_k \times n_k}, k=1,2,\dots,r$, and let $\CG_{\Delta^{(k)}}$ be the directed and signed graph associated with them. 
	Assume that they satisfy the following conditions:
	
	(i) Each  $\CG_{\Delta^{(k)}}$ is acyclic and satisfies $\CG_{\Delta^{(k)}}\subseteq \BG^{(k)}_{\rm mod}$.
	
	(ii) The matrix $S=[s_{k \ell}]_{r\times r} $ defined by
	\begin{align*}
		s_{k\ell}=\Big\{\begin{matrix*}[l]
			\CR(J\PC{k}+\Delta\PC{k})-\bar\gamma_{k k} , &\text{ if }k=\ell,\\
			-\bar \gamma_{k \ell}, &\text{ if }k \neq\ell.
		\end{matrix*}
	\end{align*}
	is an $M$-matrix.
	
	Then, there exists $\varepsilon_0>0$ such that for any $\varepsilon<\varepsilon_0$, vibrational control inputs, which functionally change the Jacobian matrix $J\PC{k}$ in each cluster  to $J\PC{k}+\Delta\PC{k}, k=1,2,\dots, r$,  stabilize the cluster synchronization manifold $\CM$. 
\end{corollary}

We wish to mention that one can follow the same steps as in Sections~\ref{sec:robustness}-\ref{subsec:linear} and \ref{subsec:kuramo} to design the amplitudes and frequencies of vibrational inputs. 

\section{Numerical Study}\label{numerical}
In this section, we employ an example to show how to design vibrations to stabilize cluster synchronization in a Kuramoto-oscillator network. 

The network we consider is shown in Fig.~\ref{fig:kuramoto_stabilization}~(a). Partitioning the network into two clusters $\CC_1$ and $\CC_2$, Assumption~\ref{invariance} is satisfied so that the corresponding cluster synchronization manifold $\CM$ is invariant. However, this pattern of cluster synchronization is unstable (see in Fig.~\ref{fig:kuramoto_stabilization}~(b)). Then, we design a vibrational control to stabilize it.

We observe that the first cluster has the same network structure as that in Fig.~\ref{design:kuramoto}. Within each cluster, one can derive that the linearized system $\dot x\PC{k}= J^{(k)} x\PC{k}$ has
\begin{equation*}
	J^{(1)}=\alpha\begin{bmatrix}
		-8 & 0 & 2\\
		-1 & -4 & -1\\
		1 & -1 & -5\\
	\end{bmatrix}, \quad
	J^{(2)}=\begin{bmatrix}
		-3 & 0 & 1\\
		-1 & -2 & 1\\
		1 & 0 & -3\\
	\end{bmatrix},
\end{equation*}
where $\alpha=0.05$. Following the definition \eqref{measure:robust}, one can compute that $\CR(J^{(1)})=0.305$ and $\CR(J^{(2)})=3.62$. One notices that the robustness of synchronization within $\CC_1$ is very small. Therefore, our objective becomes to improve the robustness in $\CC_1$. Since cluster $\CC_1$ share the same network structure as that in Fig.~\ref{design:kuramoto}, we can use the realizable graphs and changes identified there. Particularly, we choose 
\begin{equation*}
	\Delta\PC{1}=0.05 \begin{bmatrix}
		0 & 1& 0\\
		0 &0 &0\\
		-1 &0 & 0
	\end{bmatrix}
\end{equation*}
as the change we want to bring to $\CC_1$. 

 \begin{figure}[t]
	\centering
	\includegraphics[scale=1]{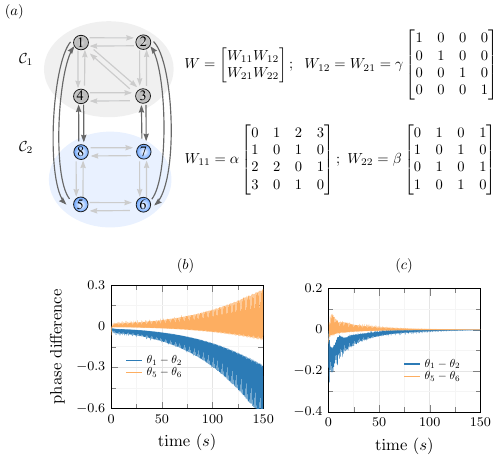}
	\caption{Vibrational stabilization of a cluster synchronization manifold. (a) The network structure and the connection weights, where $\alpha=0.05$, $\beta=1$, and $\gamma=3$ (b) Phase differences without control, indicating that the cluster synchronization is unstable. (c) Phase differences under vibrational control to the cluster $\mathcal{C}_1$. showing that the cluster synchronization has been stabilized by local vibrations. The natural frequencies in $\mathcal{C}_1$ and $\mathcal{C}_2$ are $\omega_1=1$ and $\omega_2=10$, respectively.}
	\label{fig:kuramoto_stabilization}
\end{figure}

One  can compute that this change improves the robustness from $\CR(J^{(1)})=0.305$ to $\CR(J^{(1)}+\Delta\PC{1})=0.332$.  

 Following the steps in Theorem~\ref{vibrational:design}, to realize the changes described by $\Delta^{(1)}$, we inject the vibrations below to the connections $a_{21}, a_{31}, a_{14}$, and $a_{45}$, respectively:

\begin{align*}\bl
	&u_{21}(t)=-u_{31} (t) = \frac{k_1}{\varepsilon}\sin\big(\frac{\beta_1t}{\varepsilon}\big),\\
	&u_{14}(t)=-u_{43} (t) = \frac{k_2}{\varepsilon}\sin\big(\frac{\beta_2t}{\varepsilon}\big),
\end{align*}
where $\beta_1=1$, $\beta_2=\sqrt{2}$, and
\begin{equation*}	
	k_1=\sqrt{-\frac{0.05}{J\PC{1}_{21}}}, \text{and } k_2= \sqrt{\frac{0.1}{J\PC{1}_{13}}}.
\end{equation*}

Now, according to Theorem~\ref{theorem:general}, there exist a threshold for $\varepsilon$ such that $\varepsilon$ needs to be less than it. Yet, how to identify that threshold is still an open problem. In practice, one can simply try to choose a small $\varepsilon$. Fortunately, $\varepsilon$ often does not need to be very small. For instance, as it is shown in Fig.~\ref{fig:kuramoto_stabilization}, $\varepsilon=0.01$ is sufficient to stabilize the cluster synchronization.

We wish to mention that the condition in Theorem~\ref{theorem:general} and Corollary~\ref{coro} are not even satisfied. This indicates that the condition we have identified is still a bit conservative.  More tight conditions call for future studies. However, it is worth highlighting the power of vibrational control since a slight improvement on the robustness effectively stabilizes the cluster synchronization.

\section{Discussion}
Cluster synchronization plays an very important role in many natural and man-made systems. Losing of stability of desired patterns of cluster synchronization often means malfunction. In this paper, we study how vibrational control, an open-loop control strategy, can be used to stabilize cluster synchronization. We construct some sufficient conditions under which a vibrational control stabilizes cluster synchronization. We show that one of the working mechanisms of vibrational control is that it improves the robustness of local synchronization within each cluster. Further, we provide a tractable approach to design vibrational control inputs.  We utilize numerical experiments to valide our theoretical findings. 

We conjecture that vibrational control can provide some interpretation of how deep brain stimulation works, which can potentially inform the design of better brain stimulation therapies. In the future, we are interested in extending the existing open-loop vibrational control strategy to closed-loop ones, even in the presence of imperfect measurement of system states. This hopefully would contribute to the design of better closed-loop deep brain stimulation.

\appendix

\subsection{Derivation of the Compact System}\label{derivatin}

Without loss of generality, we arrange the columns in the incidence matrices, as defined in Section~\ref{prelim}~\ref{graph:notation}, in a manner such that the intra-cluster edges within the subnetwork precede the inter-cluster edges. Mathematically, we have

\begin{equation*}
	\begin{matrix*}[l]
				B=[B_\intra, B_\inter],&B_\intra=\blkdiag( B^{(1)}_{\intra},\dots, B^{(r)}_{\intra})\\
				\hat B=[\hat B_\intra, \hat B_\inter],&\hat B_\intra=\blkdiag(\hat B^{(1)}_{\intra},\dots,\hat B^{(r)}_{\intra}).
		\end{matrix*}
\end{equation*}

Recall that $w_{ij}$'s are the connection weights in the network $\CG$. Now, we define the diagonal matrix consisting of these weights by
\begin{equation}\label{dia:weights}
	\BmW=\begin{bmatrix}
			\BmW_\intra&0\\
			0&\BmW_\inter
		\end{bmatrix}\in \R^{m \times m}
\end{equation}
where $\BmW_\intra:=\diag\{w_{ij},(i,j)\in \CE_\intra\}$ and $\BmW_\inter:=\diag\{w_{ij},(i,j)\in \CE_\inter\}$ are also diagonal matrices, containing the weights of intra- and inter-cluster connections, respectively. Likewise, one can use 
\begin{equation}
	\BmV(t)=\begin{bmatrix}
		\BmV_\intra(t)&0\\
		0&\BmV_\inter(t)
	\end{bmatrix}
\end{equation}
to denote the diagonal matrix that contains the vibrations injected to the corresponding edges in \eqref{dia:weights}.  Note that the bold notation $\BmW$ and $\BmV$ are  different from $W$ and $V$ defined in Section~\ref{no_input}; in fact, they are obtained by rewriting the non-zeros entries in $W$ and $V$ into a diagonal matrix, respectively. 
Now, one can rewrite the controlled system \eqref{Kuramoto:controlled} into
\begin{equation}
	\dot \theta = \omega - \BB (\BmW+\BmV(t)) \sin(B^\top \theta), 
\end{equation}
where $\omega=[\omega_1,\omega_2,\dots, \omega_n]^\top$, $\BB$ is a matrix obtained by replacing the negative elements in the oriented incidence matrix $B$ with $0$.  Likewise, we define $\BB_\intra$ and $\BB_\inter$, and then we have $\BB=[\BB_\intra,\BB_\inter]$. Also, we write $\BB_\intra=[\BB_\intra\PC{1},\BB_\intra\PC{2},\dots,\BB_\intra\PC{r}]$.

Recall that $x=\hat B_\intra^\top \theta$ and $y=\hat B_\inter^\top \theta$. Then, it holds that 
\begin{subequations}\label{derive:intermediate}
\begin{align}
	&\dot x = -  \hat B_\intra^\top \BB (\BmW+\BmV(t)) \sin(B^\top \theta),\\
	&\dot y = \hat B_\inter^\top \omega - \hat B_\inter^\top \BB (\BmW+\BmV(t)) \sin(B^\top \theta), 
\end{align}
\end{subequations}
where the fact that identical intra-cluster natural frequencies imply $\hat B^\top_\intra \omega=0$ has been used.
Now, it remains to write $B^\top \theta$ into a function of $x$ and $y$. 

To this end, we provide the following instrumental lemma. Below, $\bar B$ is the incidence matrix of the undirected counterpart of $\CG$ (replacing every pair of bidirectional edges with one undirected edge); and $\bar B_\inter$ and $\bar B_\intra$ are also the undirected counterparts of $ B_\inter$ and $ B_\intra$, respectively. The matrix $\bar B_\intra$ can be decomposed with respect to different clusters as $\bar B_\intra=[\bar B_\intra^{(1)},\dots, \bar B_\intra^{(r)}]$. Also, we define two projection matrices
\begin{equation*}
	P_\intra = I_n-\hat B_\intra \hat B_\intra^\dagger, P_\inter=I_n-\hat B_\inter \hat B_\inter^\dagger,
\end{equation*}
where $(\cdot)^\dagger$ denotes the pseudo-inverse of a matrix. 

\begin{lemma}\label{incidence:transfer}
	For the incidence matrices $B\in \R^{n \times m}$ and $\hat B\in \R^{n \times (n-1)}$ of the graph $\CG$ and its directed spanning tree $\hat {\CG}$, there exists 
	\begin{equation*}
		R= \begin{bmatrix}
		R_1 & \BmZr\\
		R_2 &R_3
		\end{bmatrix}
	\end{equation*}
	such that $B^\top =R\hat B^\top$,  where 
	\begin{align*}
		R_1= \begin{bmatrix}
			R_1'\\-R_1'
		\end{bmatrix}, R_2= \begin{bmatrix}
		R'_2\\-R_2', 
		\end{bmatrix},R_3= \begin{bmatrix}
		R'_3\\-R_3'
		\end{bmatrix},
	\end{align*}
	with 	
	\begin{align*}
		&R'_1= \bar B_\intra^\top (\hat{B}_\intra^\top P_\inter)^\dagger,\\
		&R'_2=\bar B_\inter^\top (\hat B_\intra ^\top P_\inter)^\dagger,\\
		&R'_3=\bar B_\inter^\top (\hat B_\inter ^\top P_\intra)^\dagger.  \hspace{2cm}\QEDA
	\end{align*} 
\end{lemma}

\begin{proof}
	As $\bar B$ is the incidence matrix of the undirected counterpart of $\CG$, we can write 
	\begin{equation}\label{Lem4:1}
		B = [\bar B_\intra, -\bar B_\intra, \bar B_\inter, - \bar B_\inter].
	\end{equation}
	The work \cite{KR-IH:2021} has shown that 
	\begin{equation*}
		\bar B^\top=\begin{bmatrix}
			\bar B_\intra ^\top \\
			 \bar B_\inter ^\top
		\end{bmatrix}=\begin{bmatrix}
		R'_1 & \BmZr\\
		R'_2 &R'_3
		\end{bmatrix} \hat B^\top,
	\end{equation*}
	which means that 
	\begin{equation*}
		\bar B_\intra ^\top = R'_1 \hat B^\top, \text{ and }\bar B_\inter ^\top= R'_2 \hat B_\intra^\top + R'_3 \hat B_\inter^\top .
	\end{equation*}
	Substituting them into \eqref{Lem4:1} completes the proof.		
\end{proof}

Applying Lemma~\ref{incidence:transfer} to Eq. \eqref{derive:intermediate}, we obtain
\begin{equation*}
	B^\top \theta= \begin{bmatrix}
		R_1 & 0\\
		R_2 &R_3
	\end{bmatrix}\hat B^\top \theta= \begin{bmatrix}
	R_1 & 0\\
	R_2 &R_3
	\end{bmatrix} \begin{bmatrix}
	x\\ y
	\end{bmatrix}.
\end{equation*}

Subsequently, one can derive that the functions in the compact system \eqref{compact_form} are given by
\begin{subequations}\label{expr:function}
\begin{align}
	&f_\intra (x) = -\hat B_\intra^\top \BB_\intra \BmW_\intra \sin (R_1x) \\
	&f_\inter (x,y)= - \hat B_\intra^\top \BB_\inter \BmW_\inter \sin(R_2 x +R_3y) \label{expr:inter-pert}\\
	&h_{\ctl}(V(t),x,y)=- \hat B_\intra^\top \BB_\intra \BmV_\intra(t) \sin (R_1x) \nonumber\\
	& \hspace{1.2cm} - \hat B_\intra^\top \BB_\inter \BmV_\inter(t) \sin(R_2 x +R_3y), \\
	&g(x,y)=\hat B_\inter^\top \omega - \hat B_\inter^\top \BB_\intra \BmW_\intra \sin (R_1x) \\
	&\hspace{1.4cm} - \hat B_\inter^\top \BB_\inter \BmW_\inter \sin(R_2 x +R_3y) \\
	&h'_\ctr(V,x,y)= - \hat B_\inter^\top \BB_\intra \BmV_\intra(t) \sin (R_1x) \nonumber\\
	& \hspace{1.4cm} - \hat B_\inter^\top \BB_\inter \BmV_\inter(t) \sin(R_2 x +R_3y).
\end{align}
\end{subequations}

Similar to $\BmV(t)$, one can define $\BmU(t):=\diag(u_{ij}(t):(i,j)\in \CE)$, and it follows that $\BmV(t) = \frac{1}{\varepsilon}\BmU(\frac{t}{\varepsilon})$. Then, \eqref{expr:function} defines the functions in \eqref{compact_form}. 

The function $f_\intra$ characterizes the inherent dynamics occurring within the clusters, while $f_\inter$ accounts for the dynamics resulting from the inter-cluster connections. The functions $h_{\ctl}$ and $h'_{\ctl}$ delineate the impact on the dynamics introduced by the vibrational control inputs.

In this paper, we specifically investigate a form of vibrational control where vibrations are only introduced to the intra-cluster connections. Then, we have $\BmV_\inter(t)=0$. As a consequence, 
\begin{align*}
	&h_{\ctl}(V(t),x,y) = \hat B_\intra^\top \BB_\intra \BmV_\intra(t) \sin (R_1x),\\
	& h'_\ctr(V,x,y)= - \hat B_\inter^\top \BB_\intra \BmV_\intra(t) \sin (R_1x),
\end{align*}
which no longer depend on $y$. Therefore, we denote them as $h_\ctr(\BmV(t),x)$ and $h'_\ctr(\BmV(t),x)$ for notational simplicity. 

\subsection{Proof of Lemma~\ref{lemma:linearized}}\label{proof:linearized}
Since $x=0$ is exponentially stable uniformly in $y$ for the system \eqref{linearized}, according to the converse Lyapunov theorem  (see \cite[Th. 4.4]{haddad2011nonlinear} and \cite{QY-KY-BDOA-CM:2021}) there exists $\CD=\{\bar x\in\R^{n-r}:\|\bar x\|\le \rho_1\}$ and a  continuously
differentiable function function $V:[0,\infty]\times \CD\times \R^r\to \R $ such that 
\begin{align*}
			\frac{\partial V}{\partial t}+\frac{\partial V}{\partial \bar x}\big( (J + P(t)) \bar x+N(y)\bar x\big)\le -c_1\|\bar x\|^2
\end{align*}
and $\|\frac{\partial V}{\partial x}\|\le c_2\|\bar x\|$ for some constants $c_1,c_2> 0$.  
Let 
\begin{equation*}
	h(t,x)=f_\intra(x)+f_\ctr(U(t),x)+f_\inter(x,y)
\end{equation*}
and 
$
\Delta(t,x)=h(t,x)-(J + P(t))x-N(y)x.
$
It can be checked that $\partial h/\partial x$ is bounded and Lipschitz on $\CD$. Then, similar to the proof of \cite[Th. 4.13]{HKK:02-bis}, one can show that $\|\Delta(t,x)\|\le c_3 \|x\|^2$ for some $c_3>0$. The time derivative along the system  \eqref{compact_form} satisfies
\begin{align*}
	&{\frac{\partial V}{\partial t}+\frac{\partial V}{\partial x}\Big( \big(J + P(t)+N(y) \big)x+\Delta(t,x)\Big)}\\
	&\le -c_1\|x\|^2+c_2c_3\|x\|^3\\
	&\le -(c_1-c_2c_3\rho)\|x\|^2, \forall \|x\|< \rho.
\end{align*}
Choosing $\rho=\min\{\rho_1,c_1/c_2c_3\}$ completes the proof. 

\subsection{Proof of Lemma~\ref{lemma:inc_dec}} \label{pf:inc_dec}

To construct the proof, one just needs to find vibrational control inputs that increase/decrease the weight of the edge $(i,j)$ functionally in both situations.

Consider a vibrational control that is only injected to the edge $(i,j)$, i.e., $V(t)$ in \eqref{controlled_net_compact} satisfies $v_{pq}(t)=0$ for any $(p,q)\neq (i,j)$ and $V_{ji}(t)\neq 0$. One can label the nodes in the expanded network $\CG$ such that $i=1$ and $j=n$. Then, the vibrational control matrix becomes 
\begin{equation*}
	V(t)=\begin{bmatrix}
		0&0&\cdots&0 \\
		0&0&\cdots&0\\
		\vdots&\vdots&\ddots&\vdots \\
		v_{n1}(t)&0&0&0 
	\end{bmatrix},
\end{equation*}
which has a quasi-lower-triangular form. Following the steps in \cite{SMM:80}, one can derive that $\bar A$ in the averaged system \eqref{linear:averaged} is
\begin{equation*}
	\bar A = A + \bar B, \text{ where } \bar B = \begin{bmatrix}
		0&0&\cdots&0 \\
		0&0&\cdots&0\\
		\vdots&\vdots&\ddots&\vdots \\
		b_{n1}(t)&0&0&0 
	\end{bmatrix},
\end{equation*}
where $b_{n1} = -a_{1n} \lim_{T \to \infty}\frac{1}{T} \int_{t=0}^{T} F^2_{n1}(t) dt$ with $F_{n1}(t)=\int_{0}^{t}v_{n1}(\tau) d\tau$. If the edge $(n,1)$ has a positive weight, i.e., $a_{1n}> 0$, $b_{n1}<0$. If $a_{1n}<0$, we have $b_{n1}>0$,  which completes the proof.

\subsection{Proof of Theorem~\ref{theorem:general}} \label{proof:general}
One can observe that (i) implies (ii). Then, it suffices to prove the exponential stability of $x=0$ for \eqref{change:time:scale}. To do that, we first present the following lemma. 
\begin{lemma}[Growth bound of perturbations]\label{lemma:bound:pert}
	There exist some constants $\bar \gamma_{k\ell}>0$, $k,\ell=1,\dots,r$, such that, for any $k$, it holds that 
	$
	\left\|G\PC{k}(y) z  \right\|\le \sum_{\ell=1}^{r}\bar \gamma_{k\ell} \|z_\ell\|.
	$\QEDA
\end{lemma}

\begin{proof}
	To construct the proof, we write $N(y)$ into a block-diagonal form
	\begin{equation*}
		N(y)=\begin{bmatrix}
			N^{11}(y)&\dots& N^{1r}(y)\\
			\vdots& \ddots & \vdots\\
			N^{r1}(y)&\dots& N^{rr}(y)
		\end{bmatrix},
	\end{equation*}
	where each $N^{ii}(y)\in \R^{(n_i-1)\times (n_i-1)}$. Then, it follows that 
	\begin{equation*}
		G\PC{k}(y) z = \sum_{\ell =1}^{r} \big(\Phi^{-1}\big)^{(k)} N^{(k\ell)}(y ) \Phi^{(\ell)} z_{\ell}.
	\end{equation*}
	Recall that $\Phi$ and $\Phi^{-1}$ are almost periodic, $\|\Phi\|$ and $\|\Phi^{-1}\|$ are both bounded. Let $c=\max_{k,\ell}\|\big(\Phi^{-1}\big)^{(k)}\| \cdot \|\Phi^{(\ell)}\|$. Then, one can derive that
	\begin{equation*}
		\left \| G\PC{k}(y) z \right \| \le c \sum_{\ell =1}^{r} \|N^{(k\ell)}(y ) z_\ell\| \le c \sum_{\ell =1}^{r} \|N^{(k\ell)}(y )\| \|z_\ell\|. 
	\end{equation*}
	Now, it remains to bound $\|N^{(k\ell)}(y )\|$. Recall that $N(y)$ is given by 
	\begin{equation*}
		N(y)=-(\hat B_\intra)^\top \BB_\inter W_\inter (\BONE_{r-1}\otimes\sin(R_3y))\odot R_2.
	\end{equation*}
	Since $\sin(R_3)y$ is bounded independent of $y$, there exists a constant $c_{k\ell}$ such that $\|N^{k\ell}(y)\|\le c_{k\ell}$. Letting $\bar \gamma_{k\ell}=c c_{k\ell}$ completes the proof.
\end{proof}

Let $\scalemath{0.85}{V_k=z_k^\top \bar X_kz_k}$, and we have 
\begin{equation*}
	{\|\partial V_k/\partial z_k \| \le \lambda_{\max}( \bar X_k)\|z_k\|}.
\end{equation*}
	 Choose $\scalemath{0.85}{V(z)= \sum\nolimits_{k=1}^{r}d_k V_k}$ as a Lyapunov candidate. The time derivative of $V(z)$ satisfies
\begin{align*}
	{\dot V(z)}&=\sum_{k=1}^{r} d_k[z_k^\top((\bar J\PC{k})^\top  \bar X_k+  \bar X_k \bar J\PC{k})z_k \\
		&+\frac{\partial V}{\partial z_k} \Phi^{-1} N(y) \Phi z_k]\\
	&{\le \sum_{k=1}^{r} d_k[-\|z_k\|^2 + \lambda_{\max}( \bar X_k)\sum_{\ell =1}^{r}\bar \gamma_{k \ell} \|z_k\|\|z_\ell\| ]},
\end{align*}
where the inequality has used Lemma~\ref{lemma:bound:pert}.

Let ${D:=\diag(d_1,\dots,d_r)}$ and ${\hat S=[\hat s_{ij}]_{r\times r}}$ where 
\begin{align*}
	\hat s_{k\ell}=\Big\{\begin{matrix*}[l]
		1- \lambda_{\max}( \bar X_k) \bar \gamma_{k k} , &\text{ if }k=\ell,\\
		-\lambda_{\max}( \bar X_k)\bar \gamma_{k \ell}, &\text{ if }k \neq\ell.
	\end{matrix*}
\end{align*}
Then, one can rewrite ${\dot V(z)\le -\frac{1}{2}z^\top (DS+S^\top D)z}$. By assumption, $S$ is an $M$-matrix, and so is $\hat S$ since $\hat S=S \cdot\diag(\lambda_{\max}( \bar X_1),\dots,\lambda_{\max}( \bar X_r))$. It follows from \cite[Th. 9.2]{HKK:02-bis} that the system \eqref{average} is exponentially stable.  

Following similar steps as in \cite[Th. 10.4]{HKK:02-bis}, one can prove that there exists $\varepsilon^*>0$ such that for any $\varepsilon<\varepsilon^*$, $z=0$ is exponentially stable uniformly in $y$ for the system \eqref{coordinated}. Since $x(s)=\Phi(s,s_0)z(s)$ and $\|\Phi\|$ is bounded, then $x=0$ is also exponentially stable uniformly in $y$ for \eqref{change:time:scale}, which completes the proof.

\bibliographystyle{IEEEtran}
\bibliography{\alias,\FP,\Main,\New}

\end{document}